\documentclass[11pt]{article}  
\usepackage{fullpage} 
\usepackage{amsthm}
\usepackage{amsmath}
\usepackage{amssymb}
\usepackage{todonotes}
\usepackage{mathtools}
\usepackage{xspace}
\usepackage{xcolor}
\usepackage{verbatim}
\usepackage{enumitem}
\usepackage{multirow}
\usepackage{nicefrac}
\usepackage{hyperref}
\usepackage{cite}

\usepackage{amsfonts}
\usepackage{bbold}
\usepackage{todonotes}
\usepackage{thm-restate}

\newtheorem*{theorem*}{Theorem}
\newtheorem{fact}{Fact}

\newcommand{\dist}{\mbox{\rm dist}}

\newcommand{\lpath}{\textsc{$\ell$-path of $k$-cliques}\xspace}

\newcommand{\local}{\textsf{LOCAL}\xspace}

\newcommand{\id}{\mbox{\rm id}}
\newcommand{\Ga}{$G_1^{\{\ell,k\}}$}
\newcommand{\Gb}{$G_2^{\{\ell,k\}}$}

\newcounter{mremarkctr}

\newtheorem{lemma}{Lemma}
\newtheorem{observation}{Observation}

\title{Generalizing Brooks' theorem via Partial Coloring is Hard Classically and Locally}

\author{
Jan Bok\thanks{Department of Algebra, Faculty of Mathematics and Physics, Charles University. Email: {\tt jan.bok@matfyz.cuni.cz}. Website: \url{https://janbok.github.io}. ORCID id: \url{https://orcid.org/0000-0002-7973-1361}.  Supported by the European Union (ERC, POCOCOP, 101071674). Views and opinions expressed are however those of the author(s) only and do not necessarily reflect those of the European Union or the European Research Council Executive Agency. Neither the European Union nor the granting authority can be held responsible for them.}
\and
Avinandan Das\thanks{Department of Computer Science, Aalto University. Email: {\tt avinandan.das@aalto.fi}. This work was supported in part by the Research Council of Finland, Grant 363558.} \\
\and
Anna Gujgiczer\thanks{MTA--HUN-REN RI Lend\"ulet ``Momentum'' Arithmetic Combinatorics Research Group, HUN-REN Alfréd Rényi Institute of Mathematics, Budapest, Hungary. Email: {\tt gujgicza@renyi.hu}. Website: \url{https://cs.bme.hu/~gujgicza/}. Supported by the Lend\"ulet ``Momentum'' program of the Hungarian Academy of Sciences (MTA).}
\and
Nikola Jedličková\thanks{Department of Applied Mathematics, 
Faculty of Mathematics and Physics, Charles University. \\  Department of Algebra, Faculty of Mathematics and Physics, Charles University. Email: {\tt jedlickova@kam.mff.cuni.cz}. ORCID id: \url{https://orcid.org/0000-0001-9518-6386}. Supported by the research grants PRIMUS/24/SCI/008, UNCE/24/SCI/022 of Charles University, and SVV–2025–260822. }
}

\date{}

\begin{document}
\maketitle

\begin{abstract}
We investigate the classical and distributed complexity of \emph{$k$-partial $c$-coloring} where $c=k$, a natural generalization of Brooks' theorem where each vertex should be colored from the palette $\{1,\ldots,c\} = \{1,\ldots,k\}$ such that it must have at least $\min\{k, \deg(v)\}$ neighbors colored differently. Das, Fraigniaud, and Ros{\'{e}}n~[OPODIS 2023] showed that the problem of $k$-partial $(k+1)$-coloring admits efficient centralized and distributed algorithms and posed an open problem about the status of the distributed complexity of $k$-partial $k$-coloring. We show that the problem becomes significantly harder when the number of colors is reduced from $k+1$ to $k$ for every constant $k\geq 3$.

In the classical setting, we prove that deciding whether a graph admits a $k$-partial $k$-coloring is NP-complete for every constant $k \geq 3$, revealing a sharp contrast with the linear-time solvable $(k+1)$-color case. For the distributed LOCAL model, we establish an $\Omega(n)$-round lower bound for computing $k$-partial $k$-colorings, even when the graph is guaranteed to be $k$-partial $k$-colorable. This demonstrates an exponential separation from the $O(\log^2 k \cdot \log n)$-round algorithms known for $(k+1)$-colorings.

Our results leverage novel structural characterizations of ``hard instances'' where partial coloring reduces to proper coloring, and we construct intricate graph gadgets to prove lower bounds via indistinguishability arguments. 
\end{abstract}
\newpage
\section{Introduction}
\label{sec:intro}


Every graph of maximum degree $\Delta$ can be properly colored with at most $\Delta+1$ colors (with a simple greedy algorithm). The famous \emph{Brooks' theorem} \cite{brooks1941colouring} states that $\Delta$ colors are enough for almost all graphs, except for two cases, complete graphs and odd cycles, where $\Delta+1$ colors are needed. This means that -- except for those two cases -- every vertex $v$ is colored from the palette $\{1,\ldots, \Delta\}$ in such a way that each of its $\deg(v)$ neighbors receives different colors than $v$.

\subsection{Partial Coloring}

A natural generalization of proper coloring is \emph{partial coloring}~\cite{BalliuHLOS19,Brandt19,DBLP:conf/spaa/Kuhn09,NaorS95}. Given two integers~$k \geq 0$ and~$c \geq 1$, a \emph{$k$-partial $c$-coloring} of a graph assigns to each vertex a color from $\{1, \dots, c\}$ such that every vertex~$v$ has at least $\min\{k, \deg(v)\}$ neighbors with a color different from its own.

This notion generalizes several well-known coloring paradigms:
\begin{itemize}
    \item For $k = \Delta$, where $\Delta$ is the maximum degree of the graph, $k$-partial $(k+1)$-coloring corresponds to proper $(\Delta + 1)$-coloring and $k$-partial $k$-coloring corresponds to proper $\Delta$-coloring.
    \item For $k = 1$, the condition captures the classical notion of \emph{weak coloring}~\cite{NaorS95}.
\end{itemize}

The problems of proper $(\Delta+1)$-coloring~\cite{AwerbuchGLP89,PanconesiS96,RozhonG20,GhaffariK21,DBLP:conf/stoc/0001G23,BarenboimEG18,FraigniaudHK16,MausT20,GG24}, proper $\Delta$-coloring~\cite{PanconesiS95, GhaffariK21, FischerHM23, Bourreau2025, GhaffariHKM18} and \emph{weak} coloring~\cite{BalliuHLOS19, Brandt19, NaorS95, DBLP:conf/spaa/Kuhn09} have been studied extensively in the distributed computing paradigm. The notion of partial coloring brings all these problems in the same framework.

For every integer $k$ and every graph~$G$, a $k$-partial $(k+1)$-coloring always exists. As outlined in~\cite{Das} and implicitly in~\cite{BalliuHLOS19}, a simple greedy algorithm achieves this: initialize all vertices with color~1, and process the vertices sequentially. For each vertex~$v$, let $C(v)$ denote the set of colors in its neighborhood. If $|C(v)| = k+1$, the vertex retains its color. Otherwise, we color it with any color from $\{1,\dots,k+1\} \setminus C(v)$. This construction guarantees that each vertex has enough differently colored neighbors at the end of the process. Moreover, once two adjacent vertices are assigned different colors during any stage of the algorithm, they permanently maintain this property.

From a distributed perspective, Das, Fraigniaud, and Ros{\'{e}}n~\cite{Das} studied the complexity of computing such a coloring in the \local model. They designed an $O(\log n \cdot \log^2 k)$-round algorithm for $k$-partial $(k+1)$-coloring using the rounding framework from~\cite{GhaffariK21} which matches their round complexity of $O(\log n\cdot \log^2\Delta)$ for proper $(\Delta+1)$-coloring. The result suggested that ``relaxing'' the constraints of proper coloring to partial coloring made the problem ``easier'', as demonstrated by the improvement of the round complexity.
They also posed an open problem; what is the distributed complexity of $k$-partial $k$-coloring? The problem of proper $\Delta$-coloring admits polylogarithmic round complexity (\cite{GhaffariK21, Bourreau2025}) in the deterministic LOCAL model. This question can also be viewed as whether this relaxation  extends computation advantage to proper $\Delta$-coloring as well.

While proper $\Delta$-coloring is well understood in the classical (centralized) setting -- with known characterizations and complexity classifications --- the same is not true for $k$-partially $k$-colorable graphs. To the best of our knowledge, no characterization or complexity classification is known for the decision problem: given a graph $G$ and integer $k$, is $G$ $k$-partially $k$-colorable? The problem is far from being trivial. For example, our quest for characterization of $k$-partial $k$-colorability quickly yielded non-trivial obstructions (which, for instance, do not even have $(k+1)$-clique as a subgraph, see Figure~\ref{fig:nontrivial-example} for such an example).

In this work, we initiate a systematic study of this problem. We first investigate the classical complexity of recognizing $k$-partially $k$-colorable graphs. Then, leveraging structural insights, we develop distributed lower bounds for computing $k$-partial $k$-colorings in graphs that are guaranteed to admit such a coloring.

For distributed computing, we work in the standard \local model of distributed computing~\cite{Linial92}. In this model, the input graph represents a communication network where each node is a computational entity. The nodes communicate in synchronous rounds, exchanging unbounded-sized messages with their neighbors. Each node has a unique identifier (typically from a polynomial-size range) and knows only its own \id \ and the size of the universe of the \id s at the beginning of the computation. The complexity of a distributed algorithm is measured in the number of rounds until all nodes produce their output. In our context, each node must eventually decide its color based on the information available in its local neighborhood.

A more detailed discussion of the \local model is provided in Section~\ref{section:distributed local model}.

\section{Related Work}

\paragraph*{Proper $(\Delta+1)$-Coloring.}

For many years, the fastest known deterministic algorithms for proper $(\Delta+1)$-coloring in graphs with maximum degree~$\Delta$ required essentially $2^{O(\sqrt{\log n})}$ rounds in $n$-node networks~\cite{AwerbuchGLP89,PanconesiS96}. This long-standing upper bound was only recently surpassed when a polylogarithmic-round algorithm for $(\Delta+1)$-coloring was introduced~\cite{RozhonG20}. All prior algorithms relied on a particular type of graph decomposition, and the breakthrough in~\cite{RozhonG20} demonstrated how to efficiently compute this decomposition within a polylogarithmic number of rounds.

A major advance followed in 2020, when it was shown that $(\Delta+1)$-coloring can be achieved in $O(\log n \cdot \log^2 \Delta)$ rounds~\cite{GhaffariK21}, this time without graph decomposition, by using rounding of a fractional solution. For large~$\Delta$, this algorithm still runs in $O(\log^3 n)$ rounds, which was subsequently improved to~$\tilde{O}(\log^2 n)$~\cite{DBLP:conf/stoc/0001G23} and more recently to~$\tilde{O}(\log^{5/3} n)$~\cite{GG24}.

The problem of proper $(\Delta+1)$-coloring has also been studied exclusively with respect to round complexity as a function of the parameter $\Delta$. This line of work has yielded algorithms for $(\Delta+1)$-coloring that run in $\widetilde{O}(\sqrt{\Delta}) + O(\log^* n)$ rounds~\cite{BarenboimEG18,FraigniaudHK16,MausT20}. However, these are efficient only when~$\Delta$ is small --- specifically, their running times are polylogarithmic in~$n$ only if~$\Delta$ itself grows polylogarithmically with~$n$.

\paragraph*{Proper $\Delta$-Coloring.}

The problem of proper $\Delta$-coloring has been extensively studied in the distributed setting. The work of~\cite{PanconesiS95} established the \emph{distributed Brooks' theorem}, showing that if all but one vertex of the graph are already $\Delta$-colored, the coloring can be extended by recoloring a path of length $O(\log_\Delta n)$ from that vertex, leading to a distributed $\Delta$-coloring algorithm with complexity $O(\Delta \cdot \text{poly} \log n)$ and a randomized version with complexity $O(\log^3 n / \log \Delta)$ rounds. Subsequent work has significantly improved these bounds. The work of~\cite{GhaffariHKM18} presented a randomized algorithm with round complexity $O(\log \Delta) + 2^{O(\sqrt{\log \log n})}$ for $\Delta \geq 4$. Later,~\cite{GhaffariK21} gave the first deterministic polylogarithmic-time algorithm for all $\Delta \geq 3$, with round complexity $O(\log^2 \Delta \log^2 n)$. The work of~\cite{FischerHM23} designed the first randomized $\text{poly}\log\log n$-round \local algorithm for the problem. Most recently, Bourreau et al.~\cite{Bourreau2025} obtained an $O(\log^4 \Delta + \log^2 \Delta \log n \log^* n)$ round deterministic \local algorithm for $\Delta \geq 3$, which remains the state of the art for the problem.

\paragraph*{Partial Coloring.}

Under certain restrictions, efficient algorithms and lower bounds for partial coloring are known. For example, in (constant) $\Delta$-regular graphs, there exists a deterministic algorithm that computes a $k$-partial 3-coloring in $O(\log^\star n)$ rounds whenever $\Delta \geq 3k - 4$ and $k \geq 3$, and a $k$-partial $k$-coloring under the condition $\Delta \geq k + 2$ and $k \geq 4$ (see~\cite{BalliuHLOS19}). The same work also establishes lower bounds: computing a 2-partial 2-coloring in $\Delta$-regular graphs requires $\Omega(\log n)$ rounds deterministically and $\Omega(\log\log n)$ rounds randomized, for any $\Delta \geq 2$. More generally, any $k$-partial $c$-coloring in $\Delta$-regular graphs with $k \geq \frac{\Delta (c-1)}{c} + 1$ also requires $\Omega(\log n)$ deterministic rounds and $\Omega(\log\log n)$ randomized rounds~\cite{BalliuHLOS19}. Additionally, Brandt~\cite{Brandt19} recently proved that 1-partial 2-coloring in odd-degree graphs cannot be solved in $o(\log^\star \Delta)$ rounds, matching the 30-year-old upper bound from~\cite{NaorS95}. Finally, both $k$-partial $O(k^2)$-coloring and $d$-defective $O({\Delta}^2/d^2)$-coloring can be computed in $O(\log^\star n)$ rounds~\cite{DBLP:conf/spaa/Kuhn09}.

\section{Our Results}

The problem of $k$-partial $k$-coloring has been studied in the restrictive settings by~\cite{BalliuHLOS19} where they showed that for constant degree  $d$-regular graphs, a $k$-partial $k$-coloring can be computed in $O(\log^* n)$ rounds provided that $d\geq k+2$ and $k\geq 4$. Contrary to these results and the existence of efficient round complexity for $k$-partial $(k+1)$-coloring in general graphs, our results show that the problem of $k$-partial $k$-coloring becomes ``hard'' in the classical (centralized) setting and ``global'' in the local distributed setting for general (constant degree) graphs for every fixed integer $k\geq 3$. In Section~\ref{section:hard instances of partial coloring}, we characterize the hard instances of partial coloring which have to be properly colored even when relaxed to partial coloring restrictions. We leverage such structure to prove the next two theorems which form the crux of our contribution.  

\paragraph*{Classical Complexity.}

We proved that determining whether a given graph is $k$-partialy $k$-colorable is NP-complete for every constant $k\geq 3$, revealing a sharp threshold between the problem of $k$-partial $(k+1)$-coloring which admits a simple $O(n)$ step greedy algorithm.

\begin{restatable}{thm}{thmNPc}
\label{theorem:NP-completeness}
For every fixed integer $k\geq3$, the problem of deciding whether an input graph $G$ is $k$-partially $k$-colorable is NP-complete. 
\end{restatable}

\paragraph*{Distributed Complexity.}

The problem of $2$-partial $2$-coloring on general graphs trivially admits an $\Omega(n)$ lower bound as $2$-partial $2$-coloring implies a proper $2$-coloring on paths which readily admits an $\Omega(n)$ round lower bound (refer to Section~7.5 of~\cite{pelegbook}).  We show that the problem of $k$-partial $k$-coloring of graphs which are $k$-partially $k$-colorable, admits an \(\Omega(n)\) round lower bound for every constant $k\geq 3$, demonstrating an exponential separation as we reduce the number of colors from $k+1$ to $k$.   
    
\begin{restatable}{thm}{thmPCLB}\label{thm:partial_coloring_lb}
	For every fixed integer $k\geq 3$ and for any $k$-partial $k$-coloring LOCAL algorithm $\mathcal{A}$ for $k$-partially $k$-colorable graphs, there exists an $n$-node $k$-partially $k$-colorable graph $G$ such that $\mathcal{A}$ requires $\Omega(n/k^3)$ rounds to color $G$.
\end{restatable}

\begin{table}[h]
\centering
\begin{tabular}{|c|c|c|c|}
\hline
\textbf{Model} & \textbf{Problem} & \textbf{Complexity} & \textbf{Source} \\
\hline
\multirow{2}{*}{Classical}
  & $k$-partial $(k+1)$-coloring & $O(n)$ time & ~\cite{Das} \\
  & $k$-partial $k$-coloring     & NP-complete & This paper \\
\hline
\multirow{2}{*}{LOCAL}
  & $k$-partial $(k+1)$-coloring & $O(\log^2 k \cdot \log n)$ rounds & ~\cite{Das} \\
  & $k$-partial $k$-coloring     & $\Omega(n)$ rounds (const.\ $k$) & This paper \\
\hline
\end{tabular}
\vspace{1em}
\caption{Complexity theoretic landscape of  $k$-partial coloring with $k$ and $k+1$ colors  in the classical and distributed LOCAL models.}
\label{tab:summary}
\end{table}


\section{Distributed \local model}\label{section:distributed local model}

The standard \textsf{LOCAL} model of distributed computing is a \emph{synchronous} model of distributed computing which was introduced in~\cite{Linial92}. The exposition of \local model here has been adapted from~\cite{Das2024}.  It comprises of an $n$-vertex communication network $G=(V,E)$. Each edge $\{u,v\}\in E$ is a bi-directional communication channel between $u$ and $v$. Each vertex gets assigned a unique identifier through the function $\id:V\rightarrow \{1,2,\ldots, n^d\}$ for some constant $d\geq 1$ which is unique to each communication network. Initially, each node only knows its identifier, the size of the \id, and potentially some input. The nodes communicate by exchanging messages along the edges of the graph. All nodes execute the same algorithm, which proceeds in synchronous rounds. In each round, every node  sends one message to each of its neighbors,  receives the messages from its neighbors, and performs some individual computation. After a certain number of rounds, every node outputs, and terminates. For example, let us consider the problem of proper $(\Delta+1)$-coloring which has been extensively studied in this model. Here, each node of the graph $G$ will initially have $\Delta$ as input and at the end of a \local coloring algorithm execution, would output the color of the node, which is a number from the palette $[\Delta+1]$.

The \local model assumes a perfect world scenario where the topology of the graph $G$ doesn't change over time and the communication channels (i.e. the edges) are perfect (i.e. the messages aren't dropped and they are delivered instantaneously). Each node $v\in V$ is a computer which is computationally unlimited. Also, there is no limit on the size of the message transmitted by the nodes.

 The complexity of a \local algorithm is measured by the number of synchronous rounds.

There is an equivalent definition of \local algorithm based on locality. Before describing the equivalent definition, we define the \emph{$t$-neighborhood} of a vertex $v$ in the graph $G$ as follows.
        \[
    \text{For $v\in V(G)$ and $t\geq 1$, }N_G^t[v] \coloneqq G[\{u\in V(G)\ \colon\ \dist_G(u,v)\leq t\}].
        \]
In other words, the $t$-neighborhood of $v$ in $G$ is the graph induced by the vertices of $G$ which are at distance at most $t$ from $v$ (in particular, $v \in N_G^t[v]$). A \emph{$t$-radius view} of $v$ is the $t$-neighborhood of $v$ along with the \id, input, and the initial state  of each  vertex in the neighborhood. 

Any $(t+1)$-round \local algorithm $\mathcal{A}$ is equivalent to each vertex computing its output independently provided given access to its $t$-radius view. Each node $v$ implements the \emph{gather algorithm} where at round $1$, node $v$  sends its \id, input, and state to each node $u\in N_G(v)$ and receives $\id(u)$, the input, and state of $u$. At the second round, it sends the collection of \id s, states, and inputs to its neighbors, collects the same from them, and computes $1$-radius view of $v$. At each round $i\geq 3$, node $v$ sends its  $(i-2)$-radius view  to each node $u\in N_G(v)$, receives $(i-2)$-radius view of  $u$, and trivially computes its $(i-1)$-radius view. Once $v$ computes its $t$-radius view, it ``simulates'' $\mathcal{A}$ offline  and compute its output. A comprehensive exposition of this equivalent definition can be found in the lecture notes~\cite{podc_allstars_ch8}.

\subsection{Proving Lower Bounds by Constructing Indistinguishable Graphs}

\sloppypar
The lower bound proofs based on indistinguishability graphs rely on the fact, that given two graphs $G_1$ and $G_2$, and assignment of \id s of their respective vertices such that there exist two vertices $v_1$ in $G_1$ and $v_2$ in $G_2$ sharing the same \id\  and $t$-radius view implies that any $t$-round \local algorithm executing on the vertices $v_1$ and $v_2$ cannot \textit{distinguish} between $G_1$ and $G_2$ and will produce the same output for for $v_1$ and $v_2$. This follows from the radius view based definition of the \local algorithm. The technique of using indistinguishability to prove \local lower bounds have been used extensively in the literature (for example, refer to ~\cite{Linial92}, section~7.5 of ~\cite{pelegbook},~\cite{KuhnMW2004}).

For designing lower bounds for our problem, we create two $n$-vertex graphs \Ga  and \Gb such that there exist two pairs of vertices $u_1,v_1 \in$ \Ga and $u_2,v_2 \in$ \Gb such that the following properties hold.
		\begin{enumerate}
			\item There exists an assignment of \id s for the vertices of \Ga and \Gb such that $u_1$ and $u_2$ receive the same \id\ and have the same $t$-radius view for every $t=o(n)$. The same holds for $v_1$ and $v_2$ as well.
			\item Any $k$-partial $k$-colorings $\chi_1$ and $\chi_2$ of \Ga and \Gb, respectively, satisfy that $\chi_1(u_1) = \chi_1(v_1)$ while $\chi_2(u_2)\neq \chi_2(v_2)$.
		\end{enumerate}

For $t=o(n)$, any $t$-round algorithm $\mathcal{A}$  executing on $u_1$, $u_2$, $v_1$ and $v_2$ cannot distinguish \Ga and \Gb and therefore, assigns the same color to $u_1$, $u_2$  and $v_1$, $v_2$ thus incorrectly coloring at least one of the two graphs. 

\section{Hard Instances of Partial Coloring}\label{section:hard instances of partial coloring}
In this section, we identify the structural properties that make a graph instance hard for $k$-partial $k$-coloring. The goal is to understand which parts of a graph force any partial coloring algorithm to behave like a proper coloring algorithm.

Given a graph~$G$ and integers~$k$ and~$c$, any algorithm for computing a $k$-partial $c$-coloring must implicitly address two subtasks:
\begin{enumerate}
    \item Identifying a subgraph~$G' \subseteq G$ which must be properly colored to satisfy the $k$-partial coloring condition --- this identification is solely dependent on $k$ and not on the number of colors~$c$.
    \item Assigning colors to $G'$ in a way that yields a proper coloring using~$c$ colors.
\end{enumerate}

The complexity of $k$-partial coloring is therefore inherently tied to the structure of this critical subgraph~$G'$. We formally characterize such subgraphs in the lemma below, which highlights the role of a specific graph parameter in determining where proper coloring becomes unavoidable.

To this end, we introduce the parameter $\Delta_E$ for a graph~$G = (V, E)$ as follows:
\begin{align*}
    \Delta_E \coloneqq \max\limits_{e=\{u,v\}\in E}\min\{\deg_G(u),\deg_G(v)\}.
\end{align*}

\begin{lemma}\label{lemma:hard_instances_of_partial_coloring}
    A $\Delta_E$-partial coloring of $G$ is a proper coloring of $G$.
\end{lemma}

\begin{proof}
    Consider an arbitrary $\Delta_E$-partial coloring $\gamma$ of $G$. For an edge $\{u,v\}\in E$, $\deg_G(u)\leq \Delta_E$  or $\deg_G(v)\leq \Delta_E$. Without loss of generality, let us assume that this holds true for vertex $u$. Therefore, by the definition of partial-coloring, each neighbor of $u$  must be colored properly with respect to $u$ by  $\gamma$, and hence, the edge $\{u,v\}$ is not monochromatic and $\gamma$ is a proper coloring of $G$.
\end{proof}

This lemma tells us that when $k = \Delta_E$, any valid $k$-partial coloring is in fact a proper coloring. Hence, the subgraph induced by edges realizing the maximum in $\Delta_E$ forms the core of difficulty for partial coloring—it forces proper coloring behavior.

These graphs thus form the \emph{hard instances} for $k$-partial $k$-coloring. As we will observe in the following sections (Observations~\ref{observation:proper_colorable} and~\ref{observation:hard_instances_for_distributed_complexity}), all our lower bound constructions for centralized and distributed settings are captured by this characterization.

\subsubsection*{Some Key Observations} 
\begin{itemize}
    \item For a given graph $G$, it is possible that  $\Delta_{E(G)} < \Delta(G)$. See Fig.~\ref{fig:nontrivial-example} for such an example. 

    In fact, the following holds: For every graph $G$, $\mathcal{K}(G)\leq \Delta_{E(G)}\leq \Delta(G)$ where $\mathcal{K}(G)$ is the degeneracy of the graph $G$. While the inequality $\Delta_{E(G)}\leq \Delta(G)$ follows immediately from the definitions, here is a simple argument that $\mathcal{K}(G)\leq \Delta_{E(G)}$. Partition $V(G)$ into two sets: $H = \{v\in V(G)\mid\deg(v)> \Delta_{E(G)})\}$ and $L = V(G)\smallsetminus H$. The set $H$ is an independent set: if $u, v \in H$ were adjacent, the vertices $u$ and $v$ would be incident to at least $\Delta_{E(G)}+1$ 
edges, contradicting the definition of $\Delta_E$. Now, order the vertices so that all vertices of $H$ appear first (in an arbitrary order), followed by all vertices of $L$ (also in an arbitrary order). In this ordering, each vertex has at most $\Delta_{E(G)}$ neighbors appearing later, and thus $\mathcal{K}(G) \leq \Delta_{E(G)}$.
    
 Therefore, the problems of $k$-partial $(k+1)$-coloring and $k$-partial $k$-coloring are generalizations of proper $(\Delta_E+1)$-coloring and 
proper $\Delta_E$-coloring, respectively. These in turn are non-trivial generalizations\footnote{
    A problem $\Pi$ is a \emph{generalization} of a problem $\Pi'$ 
if every instance of $\Pi'$ can be expressed as an instance of $\Pi$ such that 
any algorithm for $\Pi$ also solves $\Pi'$ with the same complexity.
    }  of proper $(\Delta+1)$-coloring and proper $\Delta$-coloring. 
    \item The $k$-partial $(k+1)$-coloring algorithm of~\cite{Das} implicitly provides  an algorithm for $(\Delta_E+1)$-coloring in $O(\log^2\Delta_E\cdot \log n)$ rounds by adapting the Ghaffari-Kuhn rounding framework~\cite{GhaffariK21}. Our results show that reducing even one color from the palette makes the problem ``hard'' in the classical model as well as in the LOCAL model.   
  
\end{itemize}

\section{Classical Complexity}\label{section:classical_complexity}

\begin{figure}[t]
    \centering
    \includegraphics[scale=1.8]{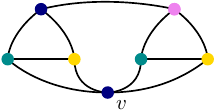}
    \caption{A nontrivial obstruction for $3$-partial $3$-coloring, showing an example of a graph with $\Delta_E = 3$ which does not admit a proper $3$-coloring, and therefore, a $3$-partial $3$-coloring. Also, observe that for this instance,  $\Delta = 4$ and therefore, $\Delta_E<\Delta$.}
    \label{fig:nontrivial-example}
\end{figure}

This section is dedicated to proving the following theorem.
\thmNPc*

We prove Theorem~\ref{theorem:NP-completeness} by reduction from the decision problem of  proper $k$-colorability of an input graph. It was shown to be NP-compelete for every constant $k\geq 3$~\cite{GareyJohnson1979}. Given an input graph $G$ of proper $k$-colorability, we transform $G$ into $G'$ by replacing each edge $\{u,v\}\in E(G)$ with the \textit{edge gadget} $e_{\{u,v\}}$ which we define as follows.  

\subparagraph{Edge Gadget.} The edge gadget $e_{\{u,v\}}$  is constructed by adding additional vertices $\{(u,v,1), (u,v,2),\allowbreak\ldots, \allowbreak(u,v,k)\}$ and edges between them such that these vertices form a $k$-clique. Moreover, the vertex $u$ is made adjacent to $(u,v,1), (u,v,2),\allowbreak\ldots, \allowbreak\text{ and }(u,v,k-1)$ and the vertex $v$ is adjacent to $(u,v,k)$. See Figure~\ref{fig:edge-gadget} for an example.  

\medskip

\begin{figure}[tb]
    \centering
    \includegraphics[scale=1.8]{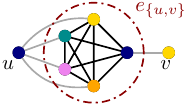}
    \caption{The edge gadget $e_{\{u,v\}}$ with $k=5$ (in red circle, connected with gray edges to the vertices $u$ and $v$). Vertices $u$ and $v$ are forced to take different colors.}
    \label{fig:edge-gadget}
\end{figure}

The transformation trivially runs in polynomial time. We now prove the following  property of the graph $G'$.

\begin{lemma}\label{lemma:coloring_end_points_of_edge_gadget}
	Any $k$-partial $k$-coloring $\gamma$ of $G'$ satisfies that $\gamma(u)\neq \gamma(v)$ for each edge $\{u,v\}\in E(G)$. 
\end{lemma}

\begin{proof}
	By the construction of the edge gadget $e_{\{u,v\}}$, for each $i\in [k]$, $\deg_{G'}((u,v,i)) = k$, therefore by the definition of $k$-partial coloring, the vertex $(u,v,i)$ must be properly colored with respect to all its neighbors. By the fact that the set $\{(u,v,1),\ldots,(u,v,k)\}$ forms a $k$-clique and that the vertex $u$ is adjacent to all vertices of the $k$-clique except $(u,v,k)$, the only possibility is that $\gamma(u) = \gamma((u,v,k))$. Since $(u,v,k)$ is adjacent to $v$ and has to be properly colored with respect to it, the claim immediately follows. 	
\end{proof}

The next claim proves the correctness of the reduction.

\begin{lemma}\label{lemma:correctness_of_NPC_reduction}
	The graph $G$ is properly $k$-colorable if and only if the graph $G'$ is $k$-partially $k$-colorable.
\end{lemma}

\begin{proof}
	Let us assume that  $G$ is properly $k$-colorable and let $\gamma$ be such a coloring of $G$. We define a coloring $\gamma'$ on $G'$ by extending the coloring of $\gamma$ on the vertices of $e_{\{u,v\}}$ for each edge $\{u,v\}$ of $G$ as follows. Assign $\gamma'((u,v,k))\coloneqq \gamma(u)$. For the rest of the vertices, assign each vertex with a unique color from the palette $[k]\smallsetminus \{\gamma(u)\}$. The resultant coloring $\gamma'$ is a proper coloring and therefore, a $k$-partial $k$-coloring of $G'$.

	Now, let us assume that $G'$ is $k$-partially $k$-colorable. Let $\gamma'$ be such a coloring of $G'$. We can deduce a proper $k$-coloring $\gamma$ of $G$ from $\gamma'$ by simply assigning $\gamma(v) \coloneqq \gamma'(v)$ for all $v\in V$. The fact that $\gamma$ is a proper coloring of $G$ follows from Lemma~\ref{lemma:coloring_end_points_of_edge_gadget}.
\end{proof}

Before we complete the proof of Theorem~\ref{theorem:NP-completeness}, we make this simple observation.

\begin{observation}\label{observation:proper_colorable}
	The graph $G'$ satisfies that $\Delta_{E(G')} = k$ and therefore, any $k$-partial $k$-coloring is also a proper $k$-coloring of $G'$.   
\end{observation}

If there exists a polynomial time algorithm for deciding $k$-partial $k$-colorability we can construct a polynomial time algorithm for $k$-colorability, which is a contradiction.
Given an instance $G$ of proper $k$-colorability, we use the transformation described above to create an instance $G'$ in polynomial time. Then we can use the $k$-partial $k$-colorability decision algorithm on $G'$. By Lemma~\ref{lemma:correctness_of_NPC_reduction}, this gives the answer whether $G$ is properly $k$-colorable or not.

\section{Distributed Complexity}\label{section:distributed_complexity}
In this section, we prove that the problem of $k$-partial $k$-coloring  is ``global''. More precisely, we prove the following theorem. 
\thmPCLB*

To prove Theorem~\ref{thm:partial_coloring_lb}, we first present the following graph gadget which has also been used in~\cite{GS16} and~\cite{DBLP:Feuilloley} in the context of distributed certification.

\paragraph*{Path of Cliques.} For integers $k\geq 3$ and $\ell\geq 2$, the graph \lpath is defined on the vertex set $V=[k]\times [\ell]$ with respect to $\ell-1$ permutations $\tau_1$, $\tau_2$, $\ldots$, $\tau_{\ell-1}\ : [k]\rightarrow [k]$. Let the set of vertices $\{(1,i),\ldots, (k,i)\}$ be denoted by $V_i$ for each $i\in \ell$.  The graph induced on $V_i$ is a $k$-clique, denoted by $K_k^i$. For each $i\in [\ell-1]$, each vertex of  $V_i$ is adjacent to every vertex of $V_{i+1}$ except the perfect matching specified by the permutation $\tau_i$, i.e. for vertices $a$ and $b$ from $V_i$ and $V_{i+1}$ respectively, $\{(a,i),(b,i+1)\}\in E$ if and only if $b\neq \tau_i(a)$. We call the set of such edges between $V_i$ and $V_{i+1}$ as \emph{anti-matching} specified by $\tau_i$. See Figure~\ref{fig:pathofcliques} for an example. Also observe that $\lpath$ is uniquely determined by the permutations $\tau_1,\ldots, \tau_{\ell-1}$.

\begin{figure}[tb]
\centering
\includegraphics[scale=1.5]{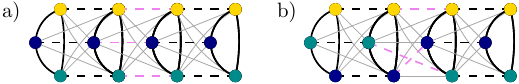}
\caption{Two $4$-path of $3$-cliques which differ only on the permutation $\tau_2$ and the effect of this difference on the coloring of the two graphs. Anti-matchings induced by $\tau_1, \tau_2, \tau_3$ are depicted with dashed lines, and the one induced by $\tau_2$ being in violet color.}      
\label{fig:pathofcliques}
\end{figure}

\begin{lemma}\label{lemma:lpath_is_k_colorable}
	For every integer $k\geq 3$ and $\ell\geq 2$,  every \lpath is properly $k$-colorable. 
\end{lemma}

\begin{proof}
	We prove this by induction on $\ell$. For $\ell =2$, consider a graph $G$ which is a \textsc{$2$-path of $k$-cliques} specified by an arbitrary injective mapping $\tau_1$  on $[k]$. A proper $k$-coloring $\chi$ of $G$ can be specified as follows. The vertices of $K_k^1$ can be colored properly with $k$ colors where each vertex is assigned a unique color. The incomplete  coloring $\chi$ can be extended to $K_k^2$ by assigning $\chi((i,2))$ the same color as $\chi(({\tau_2}^{-1}(i),1))$ for each $i\in[k]$. This is the only possible color that the vertex $(i,2)$ can take as it is adjacent to every other vertex in $K_k^1$ except  $({\tau_2}^{-1}(i),1)$.   

	 Let the statement be true for $\ell = j$ for some $j>1$. Consider a graph $G$ which is an arbitrary \textsc{$(j+1)$-path of $k$-cliques} specified by arbitrary injective mappings  $\tau_1,\ldots, \tau_{j}$ on $[k]$. By induction hypothesis, $G[K_k^1 \cup\ldots\cup K_k^j]$ can be $k$-colored properly. Let the coloring be $\chi$. It can be extended to $K_k^{j+1}$ by essentially the same arguments as in the base case. The coloring $\chi((i,j+1))$ is assigned  the same color as $\chi(({\tau_j}^{-1}(i),j))$ for each $i\in [k]$. This is the only possible color that the vertex $(i,j+1)$ can take as it is adjacent to every other vertex in $K_k^j$ except $({\tau_j}^{-1}(i),j)$. 
\end{proof}

\begin{lemma}\label{lemma:color_propagation}
	For an \lpath $G$ specified by injective mappings $\tau_1,\ldots, \tau_{\ell-1}$ and a proper $k$-coloring $\chi$ of $G$, we have that $\chi((i,1)) = \chi((\tau_{1}\circ\ldots\circ\tau_{\ell-1}(i),\ell))$ for all $i\in [k]$ and $\ell\geq 2$. 
\end{lemma}

\begin{proof}
	We prove this by induction on $\ell$. For $\ell =2$, consider a vertex $(i,2)\in K_k^2$. As there is an anti-matching between $K_k^1$ and $K_k^2$, the vertex $(i,2)$ is only not adjacent to $({\tau_1}^{-1}(i),1)$ in $K_k^1$. As $K_k^1$ is a $k$-clique, the only possible color that $(i,2)$ can take, is $\chi(({\tau_1}^{-1}(i),1))$ and hence the statement follows. Let the statement be true for $\ell = j$ for some $j>2$. Consider an arbitrary \textsc{$(j+1)$-path of $k$-cliques} specified by arbitrary injective mappings  $\tau_1,\ldots, \tau_{j}$ on $[k]$. By the same arguments as in the base case, we have that $\chi((i,j+1)) =  \chi(({\tau_j}^{-1}(i),j))$ for each $i\in [k]$. By the induction hypothesis, the claim follows. 
\end{proof}

\paragraph*{Construction of Indistinguishable Graphs.}
	 For parameters $k$ and $\ell$ such that $\ell$ is even, we define the indistinguishability  graphs \Ga and \Gb as follows. Consider two \lpath $G_1$ and $G_2$ where $G_1$ is specified by injective mappings $\tau_i = \mathbb{1}$ (where $\mathbb{1}$ represents the identity mapping) for each $i\in [\ell-1]$ and $G_2$ is specified by $\tau_i = \mathbb{1}$ for each $i\in [\ell-1]\setminus \{\ell/2\}$ and $\tau_{\ell/2}$ which is an arbitrary injective mapping on $[k]$ except $\mathbb{1}$. The graph $G_1$  is transformed to \Ga (and $G_2$ is transformed to \Gb)  by performing the graph transformation as described in the NP-completeness reduction in Section~\ref{section:classical_complexity}, i.e. replace each edge with the $k$-edge gadget. More specifically, the following transformations are performed.
 \begin{enumerate}
 	\item For $(a,b)\in [k]^2$ such that $a\neq b$ and $i\in [\ell]$, replace $\{(a,i),(b,i)\}$ with $e_{\{(a,i),(b,i)\}}$.
 	\item For $(a, b)\in [k]^2$ and $i\in [\ell-1]$ such that $b\neq \tau_i(a)$, replace  $\{(a,i),(b,i+1)\}$ with $e_{\{(a,i),(b,i+1)\}}$.
\end{enumerate} 

Both the graphs \Ga  and \Gb have $O(\ell k^3)$ vertices (as each edge of $G_1$ (and $G_2$) accounts for $k$ vertices (of the $k$-clique) and for each $i\in [\ell]$, $K_i$ has $O(k^2)$ edges and for each $j\in [\ell -1]$ the anti-matching between $V_j$ and $V_{j+1}$ has $O(k^2)$ edges). 

\begin{observation}\label{observation:hard_instances_for_distributed_complexity}
	The graphs \Ga and \Gb are properly $k$-colorable. Moreover, $\Delta_{E(\text{\Ga})} = \Delta_{E(\text{\Gb})} = k$. Therefore, any arbitrary $k$-partial $k$-coloring is also a proper $k$-coloring of \Ga. The same holds true for \Gb.	
\end{observation}

\begin{proof}
	The graphs $G_1$ and $G_2$ are properly $k$-colorable by Lemma~\ref{lemma:lpath_is_k_colorable}. As \Ga and \Gb are created by applying the same transformation on $G_1$ and $G_2$ respectively as in Section~\ref{section:classical_complexity}, the graphs \Ga and \Gb are $k$-partially $k$-colorable by Lemma~\ref{lemma:correctness_of_NPC_reduction}, properly $k$-colorable, and  $\Delta_{E(\text{\Ga})} = \Delta_{E(\text{\Gb})} = k$ by Observation~\ref{observation:proper_colorable}.
\end{proof}

\begin{fact}\label{fact:partial_coloring_lb_1}
For arbitrary $k$-partial $k$-colorings $\chi_1$ and $\chi_2$ of \Ga and \Gb respectively, the following hold:
\begin{enumerate}
	\item $\forall i\in [k],\ \chi_1((i,1)) = \chi_1((i,\ell))$ in \Ga\label{fact_indistinguishability_1},
	\item $\exists j\in [k],\ \chi_2((j,1)) \neq \chi_2((j,\ell))$ in \Gb\label{fact_indistinguishability_2}. 
\end{enumerate}  
\end{fact}

\begin{proof}
Let $\chi_1'$ and $\chi_2'$ be $k$-colorings of $G_1$ and $G_2$ respectively which are defined as follows:
\begin{align*}
	\forall i\in [k], j\in [\ell]: \chi_1'((i,j)) = \chi_1((i,j)),\\
	\forall i\in [k], j\in [\ell]: \chi_2'((i,j)) = \chi_2((i,j)).
\end{align*} 
	By the same arguments as in the proof of Lemma~\ref{lemma:coloring_end_points_of_edge_gadget}, every pair of adjacent vertices is colored properly in $G_1$ and $G_2$ and therefore,  $\chi_1'$ and $\chi_2'$ are proper colorings of $G_1$ and $G_2$, respectively. By Lemma~\ref{lemma:color_propagation}, for each $i\in [k]$, $\chi_1'((i,1)) = \chi_1'((i,\ell))$ in $G_1$ and therefore, $\chi_1((i,1)) = \chi_1((i,\ell))$ in \Ga\  as well and hence item~\ref{fact_indistinguishability_1} holds. By the same argument and the fact that $\tau_{\ell/2}\neq \mathbb{1}$, there exists $j\in [k]$ such that  $\chi_2'((j,1)) \neq \chi_2'((j,\ell))$ in $G_2$ and therefore, $\chi_2((j,1)) \neq \chi_2((j,\ell))$ in \Gb\ and item~\ref{fact_indistinguishability_2} holds.
\end{proof}
 \vspace{2mm}
\begin{fact}\label{fact:partial_coloring_lb_2}
	For each $i\in [k]$, the vertex $(i,1)$ has the same $3\big(\frac{\ell}{2}-1\big)$-neighborhood in \Ga and \Gb. The same holds true for the vertex $(i,\ell)$.  More precisely, the following holds:
	\begin{align*}
		N_{\text{\Ga}}^{3(\ell/2-1)}[(i,1)] &= N_{\text{\Gb}}^{3(\ell/2-1)}[(i,1)],\\
		N_{\text{\Ga}}^{3(\ell/2-1)}[(i,\ell)] &= N_{\text{\Gb}}^{3(\ell/2-1)}[(i,\ell)].
	\end{align*}
\end{fact}

\begin{proof}
By the construction of the graphs $G_1$ and $G_2$, we have that the graphs induced on $(V_1\cup \ldots\cup V_{\ell/2})$ by $G_1$ and $G_2$ are same. This is also true for the set $(V_{\ell/2+1}\cup\ldots V_{\ell})$. More precisely, the following holds:
\begin{align*}
	G_1[V_1\cup\ldots\cup V_{\ell/2}] &= G_2[V_1\cup\ldots\cup V_{\ell/2}],\\
	G_1[V_{\ell/2+1}\cup\ldots\cup V_{\ell}] &= G_2[V_{\ell/2+1}\cup\ldots\cup V_{\ell}].
\end{align*} 
The graphs $G_1$ and $G_2$ only differ in the anti-matching between $V_{\ell/2}$ and $V_{\ell/2+1}$. Therefore, for each $i\in [k]$, the vertex $(i,1)$ has the same $\big(\ell/2-1\big)$-neighborhood in $G_1$ and $G_2$. The same holds true for the vertex $(i,\ell)$, i.e.
	\begin{align*}
		N_{G_1}^{\ell/2-1}[(i,1)] &= N_{G_2}^{\ell/2-1}[(i,1)],\\ 
		N_{G_1}^{\ell/2-1}[(i,\ell)] &= N_{G_2}^{\ell/2-1}[(i,\ell)].
	\end{align*}

	By the construction of \Ga  and \Gb, as each edge $\{u,v\}$ is replaced consistently with the $k$-edge gadget $e_{\{u,v\}}$, and the distance between $u$ and $v$ in $e_{\{u,v\}}$ is $3$, the vertex $(i,1)$ (as well as $(i,\ell)$) has the same $3(\ell/2-1)$-neighborhood in \Ga and \Gb . 
\end{proof}
Now, we have all the tools to prove Theorem~\ref{thm:partial_coloring_lb}. 
\begin{proof}[Proof of Theorem~\ref{thm:partial_coloring_lb}]
	We prove this by contradiction. Let us assume that there exists an integer $k_0\geq 3$ and a \local algorithm $\mathcal{A}$ which $k_0$-partially colors every $n$-vertex $k_0$-partially $k_0$-colorable graph in $o(n/{k_0}^3)$ rounds. Then, there exists an integer $\ell_0>0$ such that for each even $\ell\geq \ell_0$, on every $O(\ell {k_0}^3)$-vertex $k_0$-partially $k_0$-colorable graph, the algorithm $\mathcal{A}$ terminates in at most $3\big(\ell/2 - 1\big)$ rounds.

Now let us create indistinguishability  graphs \Ga and \Gb for  $k_0$ and some even arbitrary $\ell\geq \ell_0$ as described in the section. For both \Ga and \Gb, the ids of the vertices are set as follows.
 \begin{enumerate}
	 \item For $a\in [k_0],\ i\in [\ell]$, $\id((a,i))\coloneqq (a,i)$.
 	 \item For $(a, b, j)\in [k_0]^3$ and $i\in [\ell-1]$ such that $b\neq \tau_i(a)$,\\ $\id(((a,i),(b,i+1),j)) \coloneqq ((a,i),(b,i+1),j)$.
 \end{enumerate}

 The \id s of the vertices of \Ga (and \Gb) are unique and can be encoded in $O(\log n)$ bits and therefore can be mapped uniquely to the set $\{1,2,\ldots, n^d\}$ for some constant $d$. Hence the identifiers are valid.

Let the colorings output by the algorithm $\mathcal{A}$ for \Ga and \Gb be $\chi_1$ and $\chi_2$, respectively. For each $i\in [k]$, the vertex with \id\  $(i,1)$ has the same $\big(3\ell/2-1\big)$-radius view in \Ga and \Gb (thanks to Fact~\ref{fact:partial_coloring_lb_2} and the assignment of \id s). The same is also true to the vertex with \id\   $(i,\ell)$.  Therefore, algorithm $\mathcal{A}$ at vertex with \id\ $(i,1)$ cannot distinguish between \Ga or \Gb after $\big(3\ell/2-1\big)$ rounds and outputs the same color for $(i,1)$ (and $(i,\ell)$) in \Ga and \Gb. More specifically, the following holds for vertices with \id s $(i,1)$ and $(i,\ell)$ for each $i\in[k_0]$:
		\begin{align*}
			\chi_1((i,1)) &= \chi_2((i,1)),\\
			\chi_1((i,\ell)) &= \chi_2((i,\ell)).
		\end{align*}
Item~\ref{fact_indistinguishability_1} of Fact~\ref{fact:partial_coloring_lb_1} implies that  $\chi_2((i,1)) = \chi_2((i,\ell))$ for each $i\in [k_0]$ in \Gb.  Howevere, this leads to contradiction, as due to item~\ref{fact_indistinguishability_2} of Fact~\ref{fact:partial_coloring_lb_1}, there must  exist $j\in [k_0]$ such that $\chi_2((j,1)\neq \chi_2((j,\ell))$. 
\end{proof}

\section{Conclusions}

We systematically investigated the complexity of $k$-partial $k$-coloring in both classical and distributed settings, revealing fundamental limitations of this generalized coloring paradigm. Our results demonstrate a sharp threshold between the tractable $(k+1)$-color case and the computationally hard $k$-color variant. We proved that deciding $k$-partial $k$-colorability is \textsf{NP}-complete for all constants $k \geq 3$, even for graphs where partial coloring reduces to proper coloring (as characterized by $\Delta_E$). This contrasts with the linear-time greedy algorithm for $(k+1)$-colorings.

In the \local model, we established an $\Omega(n)$-round lower bound for computing $k$-partial $k$-colorings, showing an exponential separation from the $O(\log^2 k \cdot \log n)$-round algorithms for $(k+1)$-colorings, using graph constructions based on \emph{paths of cliques} and \emph{indistinguishability arguments}. To this end, we extend the connection of this weaker version of coloring to proper coloring for proving locality as well as classical lower bounds.

\newpage

\bibliographystyle{abbrv}
\bibliography{references}

\end{document}